\documentclass[11pt]{article}

\usepackage{amsfonts,amssymb,amsmath,amsthm}
\usepackage{graphicx}
\usepackage[usenames,dvipsnames]{color}
\usepackage{hyperref}
\usepackage{multirow}
\usepackage{lineno}
\usepackage[shortlabels]{enumitem}
\usepackage[utf8]{inputenc}
\usepackage[OT4]{fontenc}
\usepackage{comment}
\usepackage{fancyhdr}
\usepackage{fullpage}
\usepackage{algorithm}
\usepackage{algorithmic}
\usepackage[usenames,dvipsnames]{xcolor}
\usepackage{tikz}

\usetikzlibrary{patterns,arrows,decorations.pathreplacing}

\newtheorem{theorem}{Theorem}
\newcommand{\minor}{\preceq}
\newtheorem{corollary}[theorem]{Corollary}
\newtheorem{lemma}[theorem]{Lemma}

\newtheorem{assumption}[theorem]{Assumption}

\newcommand{\Oof}{\mathcal{O}}
\newcommand{\CCC}{\mathcal{C}}
\newcommand{\depthone}{\text{depth-$1$ }}

\title{A new analysis of a local algorithm for the minimum dominating set problem on bounded genus graphs}

\title{A local constant factor approximation for the\\minimum dominating set problem on bounded genus graphs}

\author{Saeed Akhoondian Amiri\thanks{Technical University Berlin, Germany, 
    \texttt{saeed.amiri@tu-berlin.de}.} \quad Stefan Schmid\thanks{Aalborg University, Denmark,  \texttt{schmiste@cs.aau.dk}} \quad Sebastian
Siebertz\thanks{Technical University Berlin, Germany, \texttt{sebastian.siebertz@tu-berlin.de}}}

\date{}

\begin{document}

\maketitle

\thispagestyle{empty}
\begin{abstract}
The Minimum Dominating Set (MDS) problem is not only one of the most fundamental
problems in distributed computing, it is also one of the most challenging ones.
While it is well-known that  
minimum dominating sets cannot be approximated locally on general graphs, 
over the last years, several breakthroughs have been made on
computing local approximations on \emph{sparse graphs}. 

This paper presents a deterministic and local constant factor approximation 
for minimum dominating sets on bounded genus
graphs, a large family of sparse graphs. 
Our main technical contribution is a new analysis of
a slightly modified, first-order definable variant of an existing algorithm by Lenzen et al.
Interestingly, unlike existing proofs for planar graphs, our
analysis does not rely on any topological arguments.
We believe that our techniques can be useful for the study
of local problems on sparse graphs beyond
the scope of this paper.
 
%Our method is based on simple topological arguments which may be of independent interest. %Novelty of our work is the way that we analyse the algorithm not the algorithmical changes.

%We furthermore observe that, with some modifications and carefull analysis, the computed dominating set is first-order definable. We thereby place the constant factor approximation problem for dominating sets on bounded genus graphs in the circuit complexity class~$\mathrm{AC}^0$.
\end{abstract}

\pagebreak
\setcounter{page}{1}

\section{Introduction}\label{sec:intro}

This paper attends to the Minimum Dominating
Set (MDS) problem, arguably one of the most intensively 
studied graph theoretic problems in computer science in general,
as well as in distributed computing.

A dominating set~$D$ in a graph~$G$ is a set of vertices
such that every vertex of~$G$ either lies in~$D$ or is adjacent to a vertex in
$D$. 
Finding a minimum dominating set is NP-complete, even on
planar graphs, however, the problem can be approximated well 
on planar graphs~\cite{Baker:1994:AAN:174644.174650} 
and on classes that exclude a fixed minor~\cite{Grohe:2003:LTE:1005109.1005115}.

In this paper, we study the \emph{distributed} time complexity of
finding small dominating sets, in the classic \textit{LOCAL model} 
of distributed computing~\cite{elkin-survey,Linial:1992:LDG:130563.130578,Peleg:2000:DCL:355459,local-survey}. 
It is known that finding small dominating sets locally is hard: 
Kuhn et al.~\cite{kuhn2010local} show that in~$r$ rounds the 
MDS problem on an~$n$-vertex graphs of maximum degree 
$\Delta$  can only be approximated within factor~$\Omega(n^{c/r^2})$ and~$\Omega(\Delta^{c'/r})$, where~$c$ and~$c'$ are constants. 
This implies that, in general, to achieve a constant approximation ratio, 
every distributed algorithm requires at least~$\Omega(\sqrt{\log n})$ and~$\Omega(\log \Delta)$ communication rounds. 
Kuhn et al.~\cite{kuhn2010local} also provide the currently best approximation algorithm on general graphs, which achieves a~$(1+\epsilon)\ln \Delta$-approximation in~$\Oof(\log(n)/\epsilon)$ rounds for any~$\epsilon>0$. 

For sparse graphs, the situation is more promising.
For graphs of arboricity~$a$, Lenzen and Wattenhofer~\cite{ds-arbor}
present a forest decomposition algorithm achieving a factor~$\Oof(a^2)$ approximation
in randomized time~$\Oof(\log n)$, and a deterministic~$\Oof(a \log
\Delta)$ approximation algorithm
requiring time~$\Oof(\log \Delta)$ rounds. Graphs of bounded arboricity include all graphs which exclude a fixed graph as a (topological) minor and in particular, all planar graphs and any class of bounded genus. 
Czygrinow et al.~\cite{fast-planar} show
that given any~$\delta>0$,~$(1+\delta)$-approximations of a maximum independent
set, a maximum matching, and a minimum dominating set, can be computed in
$\Oof(\log^* n)$ rounds in planar graphs, which is asymptotically optimal~\cite{ds-alternative-lowerbound}.
Lenzen et al.~\cite{ds-planar} proposed a constant factor
approximation on planar graphs that can be computed locally in a
constant number of communication rounds. A finer analysis of
Wawrzyniak~\cite{better-upper-planar} showed that the algorithm of
Lenzen et al.\ in fact computes a~$52$-approximation of a minimum dominating set. 
In terms of lower bounds, Hilke et al.~\cite{ds-ba} show that there is no 
deterministic local algorithm (constant-time distributed graph algorithm) that 
finds a~$7-\epsilon$-approximation of a minimum dominating set on planar graphs, 
for any positive constant~$\epsilon$.

%\subsection{Our Contributions}
\textbf{Our Contributions.} 
The main contribution of this paper is 
a deterministic and local constant factor approximation 
for MDS on bounded genus
graphs. Concretely, we show that MDS can be~$\Oof(g)$-approximated
locally and deterministically on graphs of (both orientable
or non-orientable) genus~$g$.
This result generalizes existing constant factor approximation
results for planar graphs to a significantly larger graph family. 

Our main \emph{technical} contribution is a new analysis of
a sligthly modified variant of the elegant 
algorithm by Lenzen et al.~\cite{ds-planar} for planar graphs.
In particular, we introduce the natural notion of
\emph{locally embeddable graphs}: 
a locally embeddable graph 
excludes the complete bipartite graph~$K_{3,t}$ 
as a \depthone minor, that is, as minors obtained by star contractions. 
Prior works by 
Lenzen et al.~\cite{ds-planar} and Wawrzyniak~\cite{better-upper-planar} 
heavily depend 
on the planar properties of a graph: For example, their analyses
exploit the fact that each cycle in the graph
defines an ``inside'' and an ``outside'' region,
without any edges connecting the two; this facilitates a simplified
accounting and comparison to the optimal solution. 
In the case of bounded genus graphs or locally embeddable graphs,
such a global property does not exist,
and relying on non-contractible cycles can be costly: such
cycles can be very long. 
In contrast, in this paper we leverage the inherent
local properties of our low-density graphs,
which opens a new door to approach the problem. 
A second interesting technique developed in this paper
is based on \emph{preprocessing}: we show that the constants
involved in the approximation can be further improved 
by a local preprocessing step. 

%TODO: Interestingly, our analysis does not rely on any topological
%arguments. 
%Rather, our analysis is based on structural properties of graphs 
%that exclude depth-$1$ minors, that is, by minors that are obtained by 
%(simultaneous) star contractions, and on the fact that bounded genus graphs and 
%their depth-$1$ minors have bounded average degree. 

We believe that our new analysis and techniques 
can be useful also for the study
of other local problems and on more general sparse graphs, beyond
the scope of this paper.

An interesting side contribution of our modified algorithm is that
it is  \emph{first-order definable}. In particular, our 
algorithm 
does not rely on any \emph{maximum} operations, such as
finding the neighbor of maximal degree. 
The advent of
sub-microprocessor devices, such as 
biological cellular networks or networks of nano-devices, 
has recently motivated the design of very simple, ``stone-age'' distributed algorithms~\cite{stoneage},
and we believe that our work nicely complements
the finite-state machine model assumed in related work,
and opens an interesting field for future research.

%\subsection{Organization}

\textbf{Organization.} 
The remainder of this paper is organized as follows. We introduce some preliminaries in Section~\ref{sec:model}. Our basic local
approximation result is presented in Section~\ref{sec:local-approx},
and the improved approximation, using preprocessing,
is presented in  Section~\ref{sec:improved}.
After discussing a logic perspective on our work in Section~\ref{sec:stoneage},
we conclude our contribution in Section~\ref{sec:FO}.

\section{Preliminaries}\label{sec:model}

\textbf{Graphs.}
We consider finite, undirected, simple graphs. For a graph~$G$, we write
$V(G)$ for its vertices and~$E(G)$ for its edges. Two vertices~$u,v\in V(G)$ are adjacent or neighbors in~$G$ iff~$\{u,v\}\in E(G)$. The degree~$d_G(v)$ of a
vertex~$v\in V(G)$ is its number of neighbors in~$G$. 
We write~$N(v)$ for the set of neighbors and~$N[v]$ for the closed 
neighborhood~$N(v)\cup\{v\}$ of~$v$. We let~$N^1[v]:=N[v]$ and 
$N^{i+1}[v]:=N[N^i[v]]$ for~$i>1$. For~$A\subseteq V(G)$, we 
write~$N[A]$ for~$\bigcup_{v\in A}N[v]$. 
The edge density of~$G$
is~$\epsilon(G)=|E(G)|/|V(G)|$. For~$A\subseteq V(G)$, the graph
$G[A]$ induced by~$A$ is the graph with vertex set~$A$ and edge set
$\{\{u,v\}\in E(G) : u,v\in A\}$. For~$B\subseteq V(G)$ we write~$G-B$
for the graph~$G[V(G)\setminus B]$. A graph~$H$ is a subgraph of a 
graph~$G$ if~$V(H)\subseteq V(G)$ and~$E(H)\subseteq E(G)$. 

\textbf{Depth-$\mathbf{1}$ minors and local embeddable graphs.}
A graph~$H$ is a minor of a graph~$G$, written~$H\minor G$, if there is 
a set~$\{G_v : v\in V(H)\}$ of pairwise disjoint connected subgraphs 
$G_v\subseteq G$ such that if~$\{u,v\}\in E(H)$, then there is an edge 
between a vertex of~$G_u$ and a vertex of~$G_v$. We say that~$G_v$ is
contracted to the vertex~$v$ and we call~$G[\bigcup_{v\in V(H)}V(G_v)]$ 
a \emph{minor model} of~$H$ in~$G$. 

A star is a connected graph~$G$ such that at most one 
vertex of~$G$, called the center of the star, has degree
greater than one.~$H$ is a \emph{\depthone minor} of~$G$ if~$H$ is 
obtained from a subgraph of~$G$ by star contractions, that is, if 
there is a set~$\{G_v : v\in V(H)\}$ of pairwise disjoint 
stars~$G_v\subseteq G$ such that if~$\{u,v\}\in E(H)$, then 
there is an edge between a vertex of~$G_u$ and a vertex of~$G_v$.

We write~$K_{t,3}$ for the complete bipartite
graph with partitions of size~$t$ and~$3$, respectively. 
A graph~$G$ is a \emph{locally embeddable graph} if it
excludes~$K_{3,t}$ as a \depthone minor for some~$t\ge 3$. 
% Bounded-depth minors were introduced by Plotkin et al.~\cite{plotkin1994shallow} and form the basis of Ne\v{s}et\v{r}il and Ossona de Mendez's theory of nowhere dense graphs~\cite{nevsetril2012sparsity}. 
%Note that a local embeddable graph~$G$ does not necessarily have the property that every neighborhood~$N[v]$ is embeddable in particular surface. 

\textbf{Dominating set.}
Let~$G$ be a graph. A set~$M\subseteq V(G)$ \emph{dominates}~$G$ if all vertices of
$G$ lie either in~$M$ or are adjacent to a vertex of~$D$, that is, if~$N[M]=V(G)$. 
A minimum
dominating set~$M$ is a dominating set of minimum cardinality 
(among all dominating set). The size of a minimum dominating set of~$G$ 
is denoted~$\gamma(G)$. 

\textbf{$f$-Approximation.} Let~$f:\mathbb{N}\rightarrow\mathbb{R}^+$.
Given an~$n$-vertex graph~$G$ and a set~$D$ of~$G$, we say that~$D$ is 
an~$f$-approximation for the dominating set problem if~$D$ is a dominating 
set of~$G$ and~$|D| \leq f(n)\cdot \gamma(G)$. An algorithm computes an 
$f$-approximation for the dominating set problem on a class~$\CCC$ of graphs 
if for all~$G\in\CCC$ it computes a set~$D$ which is an~$f$-approximation 
for the dominating set problem. 
%If~$f$ maps every number to a fixed constant~$c$, we speak of a constant factor approximation.

\textbf{Bounded genus graphs.}
We now provide some 
background on graphs on surfaces. 
A more comprehensive discussion
can be
found in~\cite{graphsurface}.
% A surface with boundary is a Hausdorff space in which every 
% point has an open neighbourhood homeomorphic to an open 
% connected subset of the closure of the upper half-plane 
% $\mathbb{H}=\{(x,y)\in \mathbb{R}^2: y\geq 0\}$ in the 
% Euclidean plane. 
% The closed disk and the cylinder are examples of surfaces with boundary. 
% The boundary of the disc is a circle, the cylinder has two boundaries which 
% are circles. The Möbius strip is a surface with only one side and only one boundary. 
% A surface is orientable if it does not contain a homeomorphic copy of the 
% M\''bius strip. Intuitively, an orientable surface has two distinct sides. 
% Examples for orientable surfaces are the sphere and the torus, the 
% Klein bottle is non-orientable. 

% A simple curve on a surface~$\Sigma$ is a continuous mapping 
% $\gamma$ from the real interval~$[0,1]$ to~$\Sigma$ such that 
% $\gamma(a)= \gamma(b)$ for~$a,b\in [0,1]$ implies~$a,b\in\{0,1\}$.~$\gamma$ 
% is closed if~$\gamma(0)=\gamma(1)$. 

The genus of a connected surface~$\Sigma$ is the maximum number 
of cuttings along non-intersecting closed simple curves without 
making the resulting surface disconnected. We speak of the 
\emph{orientable genus} of~$\Sigma$ if~$\Sigma$ is orientable and of its 
 \emph{non-orientable genus}, otherwise.
 As all our results apply to both variants, for ease of presentation,
 and as usual in the literature, we will not mention them explicitly in the following.
%A disc and the sphere have 
% orientable genus~$0$, the torus has orientable genus~$1$. The non-orientable 
% genus of the Möbius strip is~$1$, the connected surface that results from 
% cutting along a loop on the strip is homeomorphic to a strip which is twisted 
% twice (cutting it again disconnects it). Note that this implies that there 
% are no non-orientable surfaces of genus~$0$. 

% An embedding~$\Pi$ of a graph~$G$ onto a surface~$\Sigma$ is a 
% mapping that maps from~$V(G)$ to~$\Sigma$ and from~$E(G)$ to 
% simple curves on~$\Sigma$ such that the following conditions are satisfied. 
% For an edge~$e_{uv}=\{u,v\}\in E(G)$, let~$\Pi(e_{uv})=\gamma_{uv}$. 
% \begin{enumerate}
% \item~$\Pi(v)\neq \Pi(w)$ for~$v,w\in V(G)$,~$v\neq w$, 
% \item~$\{\gamma_{uv}(0), \gamma_{uv}(1)\}=\{\Pi(u),\Pi(v)\}$ 
% for all edges~$\{u,v\}\in E(G)$, 
% \item~$\gamma_{uv}$ and~$\gamma_{wx}$ for distinct edges 
% $\{u,v\}, \{w,x\}\in E(G)$ may intersect only at their endpoints. 
% They intersect at~$\Pi(u)=\Pi(w)$ if and only if~$u=w$. 
% \end{enumerate}

The (orientable, resp.\ non-orientable) genus of a graph is the 
minimal number~$\ell$ such that the graph can be embedded on 
an (orientable, resp.\ non-orientable) surface of genus~$\ell$. 
We write~$g(G)$ for the orientable genus of~$G$ and~$\tilde{g}(G)$ for the 
non-orientable genus of~$G$. Every connected planar 
graph has orientable genus~$0$ and non-orientable genus~$1$. 
In general, for connected~$G$ we have~$\tilde{g}(G)\leq 2g(G)+1$. 
On the other hand, there is no bound for~$g(G)$ in terms of~$\tilde{g}(G)$.

We can define the genus of a surface in a combinatorial fashion as well.
Let~$G=(V,E)$ be a graph embedded on a surface~$\Pi$. 
The number~$\chi(\Pi)= |V| - |E| + |F|$ is a fixed number called the Euler 
characteristic of the embedding~$\Pi$. $V$ and $E$ are the number of vertices and edges
and $F$ is the number of faces w.r.t. $\Pi$. If~$\Pi$ is orientable then~$g(\Pi) = 1-1/2\chi(\Pi)$ 
is the orientable genus of~$\Pi$ and analogousely if~$\Pi$ is not 
orientable then~$\tilde{g}(\Pi) = 2-\chi(\Pi)$ is the non-orientable genus of~$\Pi$. 
The genus of~$\Pi$ and~$g(G)$ are equal if~$\Pi$ is a minimal surface that 
$G$ embeds on.

%In our analysis we will be contracting star subgraphs. This does not increase the genus. 
Considering that graphs of genus~$g$ are closed on 
taking subgraphs and edge contraction, we have the following lemma.
\begin{lemma}\label{lem:closureminor}
If~$H\minor G$, then~$g(H)\leq g(G)$ and~$\tilde{g}(H)\leq \tilde{g}(G)$. 
\end{lemma}

%For graphs on surfaces and embeddings we use the terminology of~\cite{graphsurface}. Let~$G=(V,E)$ be a graph with an orientable embedding~$\Pi$ and~$f$ faces. The number~$\Pi(g)=2 - |V|+ |E|- f$ is called the Euler genus of~$G$ w.r.t.\~$\Pi$. The Euler genus~$g(G)$ of a graph~$G$ is the minimum Euler genus~$\Pi(g)$ over all orientable embeddings~$\Pi$ of~$G$. Similarly we can define non-orientable euler genus~$\tilde{g}(G)$. w.r.t. Proposition 4.4.1. in~\cite{graphsurface} we always have~$\tilde{g}(G) \le g(G)$. When a graph~$G$ is clear from the context we write~$g$ or~$\tilde{g}$ respectively for~$g(G)$ ,$\tilde{g}(G)$.

% a closed surface is the maximum number of disjoint simple closed curves that can be drawn on the surface without disconnecting it. The genus of a graph~$G$ is the minimal integer~$g$ such that the graph can be drawn without crossing itself on a closed surface of genus~$g$. A planar graph is a graph of genus~$0$. We refer to the textbook~\cite{graphsurface} for background on graphs on surfaces.
%For our purpose it is sufficient to note the following two lemmas. 

One of the arguments we will use is based on the fact that bounded 
genus graphs exclude large bipartite graphs as minors 
(and in particular as \depthone minors). The lemma follows 
immediately from Lemma~\ref{lem:closureminor} and from the fact that 
$g(K_{m,n})=\left\lceil \frac{(m-2)(n-2)}{4}\right\rceil$ and 
$\tilde{g}(K_{m,n})=\left\lceil \frac{(m-2)(n-2)}{2}\right\rceil$ 
(see e.g.\ Theorem~4.4.7 in~\cite{graphsurface}). 

\begin{lemma}\label{lem:exclude}
If~$g(G)=g$, then~$G$ excludes~$K_{4g+3,3}$ as a minor 
and if~$\tilde{g}(G)=\tilde{g}$, then~$G$ excludes~$K_{2\tilde{g}+3,3}$ as a minor. 
\end{lemma}

Graphs of bounded genus do not contain 
many disjoint copies of minor models of~$K_{3,3}$: 
this is a simple consequence of the fact that the orientable 
genus of a connected graph is equal to the sum of the genera 
of its blocks and a similar statement holds for the non-orientable genus, 
see Theorem~4.4.2 and Theorem~4.4.3 in~\cite{graphsurface}.

\begin{lemma}
\label{lem:decreasegenus}
$G$ contains at most~$\max\{g(G), 2\tilde{g}(G)\}$ disjoint copies of minor models of~$K_{3,3}$. 
\end{lemma}

Finally, note that graphs of bounded genus have a small 
edge density. 
%The following lemma can easily be shown 
%by the Euler characterization and a maximal triangulation 
%without increasing its genus, and we omit the proof. 
%\stefan{saeed please add proof in appendix}

\begin{lemma}\label{lem:dens}
For any graph $G$, we always have $|E(G)| \leq 3 \cdot |V(G)| + 6 g(G) - 6$ 
and~$|E(G)| \leq 3 \cdot |V(G)| + 3 \tilde{g}(G) - 3$ 
\end{lemma}
\begin{proof}
We only prove the first inequality; 
the second can be proved analogously. Suppose
$G$ is embedded in an orientable surface of genus $g=g(\Pi)=g(G)$, 
with embedding
$\Pi$. Denote the number of faces of $G$ with respect to $\Pi$ by $f$, its number of edges by $e$ and its number of vertices by $v$. 
Every face in $\Pi$ has at least
$3$ edges and each edge appears in at most $2$ faces, so
$3f\le 2e$, and hence $f\le 2/3 e$. By the Euler formula (see above) we have: 
\begin{align*}
g & = 1-1/2\chi(\Pi)\\
& = 1-1/2(v-e+f).
\end{align*}
Hence 
\begin{align*}
2g & = -v +e - f + 2\\
& \geq e/3-v+2, 
\end{align*}
which implies $e\le 3v+6g-6$.
\end{proof}

%\begin{proof}
%Let~$S=\{K_1,\ldots,K_l\}$ be the maximum size set of minimal images of disjoint minors (depth-$1$ minors) of~$K_{3,3}$ in~$G$. Take a minimal subgraph~$G_S \subseteq G$ such that there is a path from any vertex~$v\in V(K_i)$ to a vertex~$u\in V(K_j)$ for all~$i,j \in [l]$. Blocks of~$G_S$ either are single edge or one of the~$K_i\in S$. By Lemma 4.4.3 in~\cite{graphsurface} we have.
%$g(G_S) \ge \Sigma_{i=1}^l g(K_i) = l$.
% As~$G_S \subseteq G$ the genus of~$G_S$ is at most~$g$ so~$l \le g$  
%
%\end{proof}

\textbf{Distributed complexity.}
We consider the standard \textit{LOCAL} model of distributed
computing~\cite{Linial:1992:LDG:130563.130578, Peleg:2000:DCL:355459}, 
see also~\cite{local-survey} for a recent survey. 
A distributed system is modeled as a graph~$G$. 
At each vertex~$v\in V(G)$ there is an independent 
agent/host/processor with a unique identifier~$\mathit{id}(v)$. 
Initially, each agent has no knowledge about the network, 
but only knows its own identifier. Information about other agents 
can be obtained through message passing, i.e.,
 through repeated interactions with neighboring vertices, 
 which happens in synchronous communication rounds. 
 In each round the following operations are performed:
(1) Each vertex performs a local computation 
(based on information obtained in previous rounds).
(2) Each vertex~$v$ sends one message to
 each of its neighbors.
(3) Each vertex~$v$ receives one message from 
each of its neighbors. The \emph{distributed complexity} 
of the algorithm is defined as the number of
 communication rounds until all agents terminate. 
 We call a distributed algorithm~$r$-local, if its 
 output depends only on the~$r$-neighborhoods 
~$N^r[v]$ of its vertices.

\section{The Local MDS Approximation}\label{sec:local-approx}

Let us start by revisiting the MDS approximation algorithm for planar graphs
by Lenzen et al.~\cite{ds-planar}, see Algorithm~\ref{alg:lenzen}.
The algorithm works in two phases. In the first phase, 
it adds all vertices whose (open) neighborhood cannot 
be dominated by a small number of vertices to a set~$D$. 
It has been shown in~\cite{ds-planar} that the set~$D$ is small in planar graphs. 
In the second phase, it defines a dominator function~$dom$ which 
maps every vertex~$v$ that is not dominated yet by~$D$ to its dominator. 
The dominator~$dom(v)$ of~$v$ is chosen arbitrary among 
those vertices of~$N[v]$ which dominate the 
maximal number of vertices not dominated yet.

\begin{algorithm}[h!]
\caption{Dominating Set Approximation Algorithm for Planar Graphs}
\label{alg:lenzen}
\begin{algorithmic}[1]
\vspace{2mm}
\STATE Input: Planar graph~$G$%, integer~$c\geq \epsilon(G)$
\STATE~$(*$ \emph{Phase 1} ~$*)$
\STATE~$D \gets \emptyset$
\STATE \textbf{for}~$v\in V$ (in parallel) \textbf{do}
\STATE ~~~~~\textbf{if} there does not exist a set~$A\subseteq V(G)\setminus \{v\}$ such that~$N(v)\subseteq N[A]$ and~$|A|\leq 6$ \textbf{then}
\STATE ~~~~~~~~~$D\gets D\cup \{v\}$
\STATE ~~~~~\textbf{end if}
\STATE \textbf{end for}
\STATE~$(*$ \emph{Phase 2} ~$*)$
\STATE~$D'\gets \emptyset$
\STATE \textbf{for}~$v\in V$ (in parallel) \textbf{do}
\STATE ~~~~~~$d_{V-D}(v)\gets |N[v]\setminus N[D]|$
\STATE ~~~~~ \textbf{if}~$v\in V\setminus N[D]$ \textbf{then}

\STATE ~~~~~~~~~$\Delta_{V-D}(v)\gets \max_{w\in N[v]}d_{V-D}(w)$
\STATE ~~~~~~~~ choose any~$dom(v)\in N[v]$ with~$d_{V-D}(dom(v))=\Delta_{V-D}(v)$
\STATE ~~~~~~~~~$D'\gets D'\cup \{dom(v)\}$
\STATE ~~~~~~~~ \textbf{end if}
\STATE ~~~~~ \textbf{end if}
\STATE \textbf{end for}
\STATE \textbf{return}~$D\cup D'$
\end{algorithmic}\label{alg:lenzen}
\end{algorithm}

We now propose the following small change to the algorithm. 
As additional input, we require an integer~$c\geq \epsilon(G)$ 
and we replace the condition~$|A|\leq 6$ in Line~5 by the 
condition~$|A|\leq 2c$. In the rest of this section, we show that 
the modified algorithm computes a constant factor approximation on 
any class of graphs of bounded genus. Note that the algorithm 
does not have to compute the edge density of~$G$, which is not 
possible in a local manner. Rather, we leverage Lemma~\ref{lem:dens} 
which upper bounds~$\epsilon(G)$ for any fixed 
class of bounded genus graphs: this upper bound
can be used as an input to the local algorithm. 

We first show that the set~$D$ computed in Phase~1 
of the algorithm is small. The following lemma is a straight-forward 
generalization of Lemma~6.3 of~\cite{ds-planar}, 
which in fact uses no topological arguments at all. 

\begin{lemma}\label{thm:largeneighbourhood}
  Let~$G$ be a graph and let~$M$ be 
  %\marginpar{$M$} 
  a minimum dominating set of~$G$. 
  Assume that for some constant~$c$ all \depthone
  minors~$H$ of~$G$ satisfy~$\epsilon(H)\leq c$. Let 
  \[ D:=\{v\in V(G) : \text{ there is no set~$A\subseteq V(G)\setminus\{v\}$
  such that~$N(v)\subseteq N[A]$ and~$|A|\leq 2c$}\}.\]
 Then~$|D|\leq (c+1)\cdot |M|$.
 %\marginpar{$D$}
\end{lemma}
\begin{proof}
  Let~$H$ be the graph with~$V(H)=M\cup N[D\setminus M]$ and where
 ~$E(H)$ is a minimal subset of~$E(G[V(H)])$ such that all edges with
  at least one endpoint in~$D\setminus M$ are contained in~$E(H)$ and
  such that~$M$ is a dominating set in~$H$. By this minimality
  condition, every vertex~$v\in V(H)\setminus (D\cup M)$ has exactly
  one neighbor~$m\in M$, no two vertices of~$V(H)\setminus(M\cup D)$
  are adjacent, and no two vertices of~$M$ are adjacent.

  We construct a \depthone minor~$\tilde{H}$ of~$H$ by contracting the
  star subgraphs~$G_m$ induced by~$N_H[m]\setminus D$ for~$m\in
  M\setminus D$ to a single vertex
 ~$v_m$. Let~$w\in D\setminus M$. As~$N_G(w)$ cannot be covered by
  less than~$(2c+1)$ elements from~$V(G)\setminus \{w\}$ (by definition of~$D$),~$w$
  also has at least~$(2c+1)$ neighbors in~$\tilde{H}$.  On the other hand,
 ~$\tilde{H}$ has at most~$c\cdot |V(\tilde{H})|$ edges, and also
  the subgraph~$\tilde{H}[D\setminus M]$ has at most 
 ~$c\cdot |D\setminus M|$ edges (by assumption on
 ~$\epsilon(\tilde{H})$).

  Hence
  \begin{align*}
    & \quad  (2c+1)\cdot |D\setminus M|-c\cdot |D\setminus M| \\
    & \leq  \sum_{w\in D\setminus M} d_{\tilde{H}}(w)-|E(\tilde{H}[D\setminus M])|\\
    & \leq  |E(\tilde{H})|\\
    & \leq  c\cdot |V(\tilde{H})|\\
    & = c\cdot (|D\setminus M|+ |M|),
  \end{align*}
  and hence~$|D\setminus M|\leq c\cdot|M|$, which implies the claim.
\end{proof}

\begin{assumption}
For the rest of this section, we assume that~$G$ satisfies 
that for all \depthone
  minors~$H$ of~$G$, 
$\epsilon(H)\leq c$ for some constant~$c$,
and we fix~$M$ and~$D$ as in 
Lemma~\ref{thm:largeneighbourhood}.
We furthermore assume that for some $t\geq 3$, $G$ excludes~$K_{t,3}$ as \depthone minor. 
\end{assumption}

Let us write~$R$ for the set~$V(G)\setminus N[D]$ of 
vertices which are not dominated by~$D$. The algorithm 
defines a dominator function~$dom:R\rightarrow N[R]\subseteq 
V(G)\setminus D$. The set~$D'$ computed by the algorithm
is the image~$dom(R)$ of~$d$ under~$R$. As~$R$ contains 
the vertices which are not dominated by~$D$,~$D'\cup D$ is a 
dominating set of~$G$. This simple observation proves that the 
algorithm correctly computes a dominating set of~$G$. 
Our aim is to find a bound on~$|dom(R)|$.

We fix an ordering of~$M$ as~$m_1,\ldots, m_{|M|}$ such that 
the vertices of~$M\cap D$ are first in the ordering and 
inductively define a minimal set~$E'\subseteq E(G)$ such
 that~$M$ is a dominating set with respect to~$E'$. For this, we add all 
 edges~$\{m_1,v\}\in E(G)$ with~$v\in N(m_1)\setminus M$ to~$E'$. 
 We then continue inductively by adding for~$i>1$ all 
 edges~$\{m_i, v\}\in E(G)$ with~$v\in N(m_i)
 \setminus(M\cup N_{E'}(m_1,\ldots, m_{i-1}))$. 

For~$m\in M$, let~$G_m$ be the star subgraph of~$G$ with 
center~$m$ and all vertices~$v$ with~$\{m,v\}\in E'$. Let~$H$
be the \depthone minor of~$G$ which is obtained by contracting all 
stars~$G_m$ for~$m\in M$. This construction is visualized in 
Figure~\ref{fig:construction}. In the figure, solid lines
represent edges from~$E'$, lines from~$E(G)\setminus E'$ are 
dashed. We want to count the endpoints of directed edges, 
which represent the dominator function~$dom$. 

\begin{figure}[h!]
\centering

\begin{tikzpicture}

% First star
\filldraw[draw=black,fill=black!10!white,rounded corners=10pt] (1,1) -- (0.15, -0.9) -- (1.85,-0.9) -- cycle;

\draw[fill=black] (1, 0.5) circle (1mm);
\node at (1.5, 0.8) {$m_1$};
\node at (0.1, 0.1) {$G_{m_1}$};

\foreach \x in {0.5, 1, 1.5}
{
	\draw[fill=black] (\x, -0.7) circle (0.4mm);
	\draw (\x,-0.7) -- (1,0.5);
}

% Second star

\filldraw[draw=black,fill=black!10!white,rounded corners=10pt] (3,1) -- (2.15, -0.9) -- (3.85,-0.9) -- cycle;

\draw[fill=black] (3, 0.5) circle (1mm);
\node at (3.5, 0.8) {$m_2$};

\foreach \x in {2.5, 3, 3.5}
{
	\draw[fill=black] (\x, -0.7) circle (0.4mm);
	\draw (\x,-0.7) -- (3,0.5);
}

\draw[fill=black] (4.25, -0.1) circle (0.4mm);
\draw[fill=black] (4.55, -0.1) circle (0.4mm);
\draw[fill=black] (4.85, -0.1) circle (0.4mm);

\draw[dashed] (3,0.5) -- (1.5, -0.7);

% Last star

\filldraw[draw=black,fill=black!10!white,rounded corners=10pt] (6,1) -- (5.15, -0.9) -- (6.85,-0.9) -- cycle;

\draw[fill=black] (6, 0.5) circle (1mm);
\node at (6.8, 0.7) {$m_{|M|}$};

\foreach \x in {5.5, 6, 6.5}
{
	\draw[fill=black] (\x, -0.7) circle (0.4mm);
	\draw (\x,-0.7) -- (6,0.5);
}

\draw (2.5,-0.7) edge[out=210,in=330,->] (1.58, -0.74);
\draw (5.45,-0.74) edge[out=210,in=330,<-] (3.5, -0.7);
\draw (6.4,-0.74) edge[out=210,in=330,<-] (6, -0.7);

\end{tikzpicture}
\caption{The graphs $G_m$. Solid lines represent edges from $E'$, directed edges represent the dominator function $d$.}
\label{fig:construction}
\end{figure}
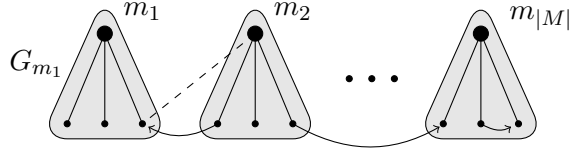

In the following, we call a directed edge which
represents the function~$dom$ a \emph{$dom$-edge}. We did not 
draw~$dom$-edges that either start or end in~$M$. When counting~$|dom(R)|$, 
we may simply add a term~$2|M|$ to estimate the number of those edges. 
We also did not draw a~$dom$-edge starting in~$G_{m_1}$. 
In the figure, we assume that the vertex~$m_1$ belongs to 
$M\cap D$. Hence every vertex~$v$ from~$N[m_1]$ is 
dominated by a vertex from~$D$ and the function is thus not 
defined on~$v$.

$H$ has~$|M|$ vertices and by our 
assumption on the 
density of \depthone minors of~$G$, 
it has at most~$c|M|$ edges.

Our analysis proceeds as follows. We distinguish between
two types of~$dom$-edges, namely those which go from 
one star to another star and those which start and end in the 
same star. By the star contraction, all edges which go 
from one star to another star are represented by a 
single edge in~$H$. We show in Lemma~\ref{lem:edgerepresentative} 
that each edge in~$H$ does not represent many such~$dom$-edges 
with distinct endpoints. As~$H$ has at most~$c|M|$ 
edges, we will end up with a number of such edges 
that is linear in~$|M|$. On the other hand, all edges 
which start and end in the same star completely 
disappear in~$H$. In Lemma~\ref{lem:insidestars}
we show that these star contractions ``absorb'' only few 
such edges with distinct endpoints.

We first show that an edge in~$H$ represents only 
few~$dom$-edges with distinct endpoints. For each
$m\in M\setminus D$, we fix a set~$C_m\subseteq V(G)\setminus \{m\}$ 
of size at most~$2c$ which covers~$N_{E'}(m)$. Recall that we 
assume that~$G$ excludes~$K_{t,3}$ as \depthone minor. 

\begin{lemma}\label{lem:edgerepresentative}
Let~$1\leq i<j\leq |M|$. Let ~$N_i:=N_{E'}(m_i)$ and~$N_j:=N_{E'}(m_j)$. 
\begin{enumerate}
\item If~$m_i\in D$ and~$m_j\not\in D$, then \[|\{u \in N_j: \text{ there is~$v\in N_i$ with~$\{u,v\}\in E(G)\}|\leq 2ct$}.\]
\item If~$m_i\not \in D$ (and hence~$m_j\not\in D$), then \[|\{u \in N_j: \text{ there is~$v\in N_i$ with~$\{u,v\}\in E(G)\}|\leq 2ct$}.\]
\item If~$m_i\not\in D$ (and hence~$m_j\not\in D$), then \[|\{u \in N_i: \text{ there is~$v\in N_j$ with~$\{u,v\}\in E(G)\}|\leq 4ct$}.\]
\end{enumerate}
\end{lemma}
\begin{proof}
By definition of~$E'$,~$m_i\not\in C_{m_2}$. 
Choose any~$c\in C_{m_j}$. There are at most~$t-1$ 
distinct vertices~$u_1,\ldots, u_{t-1}\in (N_j\cap N(c))$ 
such that there are~$v_1,\ldots, v_{t-1}\in N_i$ 
(possibly not distinct) with~$\{u_k, v_k\}\in E(G)$ for
all~$k$,~$1\leq k\leq t-1$. 
Otherwise, we can contract the star with center~$m_i$
and branch vertices~$N(m_i)\setminus \{c\}$ and
thereby find~$K_{t,3}$ as  \depthone minor, a contradiction. 
See Figure~\ref{fig:contraction} for an illustration in the
case of an excluded~$K_{3,3}$. 
% \textcolor{red}{Note however, that in the picture~$m_1$ and~$m_2$ 
% must be switched. }
Possibly,~$c\in N_j$ and
 it is connected to a vertex of~$N_i$, hence we have 
at most~$t$ vertices in~$N_j\cap N[c]$ with a connection to 
$N_i$. As~$|C_{m_2}|\leq 2c$, we conclude both the first and the second item.
 
Regarding the third item, choosing any~$c\in C_{m_i}$. If~$c\neq m_j$, 
we conclude just as above, that there are at most~$t-1$ distinct 
vertices~$u_1,\ldots, u_{t-1}\in (N_i\cap N(c))$ such that there 
are~$v_1,\ldots, v_{t-1}\in N_j$ (possibly not distinct) with 
$\{u_k, v_k\}\in E(G)$ for all~$k$,~$1\leq k\leq t-1$ and 
hence at most~$t$ vertices in~$N_i\cap N[c]$ with a 
connection to~$N_j$. Now assume~$c=m_j$. Let~$c'\in C_{m_j}$. 
There are at most~$t-1$ distinct vertices~$u_1,
\ldots, u_{t-1}\in (N_i\cap N_E(m_j))$ such that there 
are vertices~$v_1,\ldots, v_{t-1}\in N_j\cap N(c)$ (possibly not distinct) 
with~$\{u_k, v_k\}\in E(G)$ for all~$k$,~$1\leq k\leq t-1$. 
Again, considering the possibility that~$c'\in N_i$, there 
are at most~$t$ vertices in~$N_i\cap N_E(m_j)$ with a 
connection to~$N_j\cap N(c)$. As~$|C_{m_2}|\leq 2c$, 
we conclude that in total there are at most~$2ct$ vertices 
in~$N_i\cap N_E(m_j)$ with a connection to~$N_j$. 
In total, there are hence at most~$(2c-1)t + 2ct\leq 4ct$ 
vertices of the described form.  
\begin{figure}[h!]
\centering

\begin{tikzpicture}

\begin{scope}[xshift=4cm]
\begin{scope}[yscale=1,xscale=-1]
\draw[fill=black] (1, 0.5) circle (1mm);

\draw[fill=black] (1, -1.7) circle (1mm);
\node at (1.7, -2) {$c_1\in C_1$};

\foreach \x in {0.5, 1, 1.5}
{
	\draw[fill=black] (\x, -0.7) circle (0.4mm);
	\draw (\x,-0.7) -- (1,-1.7);
	\draw (\x,-0.7) -- (1,0.5);
}
\node at (0.3, -0.6) {u};
\node at (0.8, -0.6) {v};
\node at (1.7, -0.6) {w};

\filldraw[draw=black,fill=black!10!white,rounded corners=10pt] (2.4,1) rectangle (3.8,-1.2);

\draw[fill=black] (3, 0.5) circle (1mm);
\node at (1.3, 0.7) {$m_2$};
\node at (3.3, 0.7) {$m_1$};

\foreach \x in {2.7, 3.3}
{
	\draw[fill=black] (\x, -0.7) circle (0.4mm);
	\draw (\x,-0.7) -- (3,0.5);
}

\draw (2.7,-0.7) edge[out=210,in=330,-] (0.5, -0.7);
\draw (2.7,-0.7) edge[out=210,in=330,-] (1, -0.7);
\draw (3.3,-0.7) edge[out=210,in=330,-] (1.5, -0.7);

\end{scope}
\end{scope}

\node at (5, -0.7) {$\Longrightarrow$};

\begin{scope}[xshift=7cm]

\draw[fill=black] (0, 0.5) circle (1mm);
\node at (0.3, 0.7) {$m_2$};

\filldraw[draw=black,fill=black!10!white,rounded corners=5pt] (0.6,1) rectangle (1.6,0.2);
\draw[fill=black] (1, 0.5) circle (1mm);
\node at (1.3, 0.7) {$m_1$};

\draw[fill=black] (2, 0.5) circle (1mm);
\node at (2.3, 0.7) {$c_1$};

\draw[fill=black] (0, -1.5) circle (1mm);
\node at (0.3, -1.7) {$u$};

\draw[fill=black] (1, -1.5) circle (1mm);
\node at (1.3, -1.7) {$v$};

\draw[fill=black] (2, -1.5) circle (1mm);
\node at (2.3, -1.7) {$w$};

\foreach \x in {0, 1, 2}
{
	\foreach \y in {0,1,2}
	{
		\draw (\x,0.5) -- (\y, -1.5);
	}
	
}

\end{scope}

\end{tikzpicture}
\caption{Visualisation of the proof of Lemma~\ref{lem:edgerepresentative} in the case of excluded $K_{3,3}$}
\label{fig:contraction}
\end{figure}
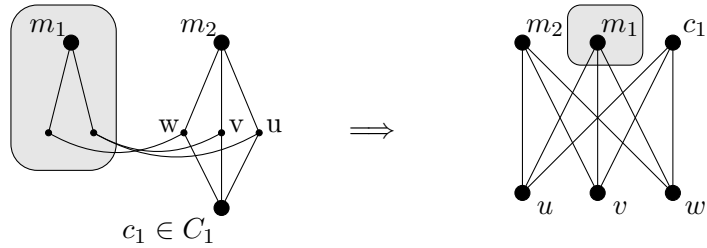
\end{proof}

%\begin{corollary}
%Let~$i\neq j$ such that~$m_j\not\in D$. Let ~$N_i:=N_{E'}(m_i)$ and~$N_j:=N_{E'}(m_j)$. Then 
%\[|\{v \in N_j: \text{ there is~$u\in N_i$ with~$d(v)=u\}|\leq 4ct$}.\]
%\end{corollary}

We write~$Y$ for the set of all vertices~$\{u\in N_{E'}(m_i) :$ 
$m_i\not\in D$ and there is~$v\in N_{E'}(m_j)$,~$j\neq i$ 
and~$\{u,v\}\in E(G)\}$. 

\begin{corollary}\label{crl:numedgesbetweendiamonds}
$|Y|\leq 6c^2t|M|$. 
\end{corollary}
\begin{proof}
Each of the~$c|M|$ many edges in~$H$ represents edges 
between~$N_i$ and~$N_j$, where~$N_i$ and~$N_j$ are defined 
as above. By the previous lemma, if~$i<j$, there are at 
most~$2ct$ vertices in~$N_i\cap Y$ and at most~$4ct$ 
vertices in~$N_j\cap Y$, hence in total, each edge accounts 
for at most~$6ct$ vertices in~$Y$. 
\end{proof}

We continue to count the edges which are inside the stars. 

\begin{lemma}\label{lem:edgestosamestar}
Let~$m\in M\setminus D$ and let~$v\in N_{E'}(m)\setminus C_m$. Then \[|\{u \in N_{E'}(m) : \{u,v\}\in E(G)\}|\leq 2c(t-1).\] 
\end{lemma}
\begin{proof}
Let~$c\in C_m$. By the same argument as in 
Lemma~\ref{lem:edgerepresentative}, there are 
at most~$t-1$ distinct vertices~$u_1,\ldots, u_{t-1}\in 
(N_{E'}(m)\cap N(c))$ such that~$\{u_k, v\}\in E(G)$ 
for all~$k$,~$1\leq k\leq t-1$. 
\end{proof}

Let~$C:=\bigcup_{m\in M\setminus D}C_m$. There are 
only few vertices which are highly connected to~$M\cup C$. 
Let~$Z:=\{u \in N_{E'}(M\setminus D) : |N(u)\cap (M\cup C)|>4c\}$. 

\begin{lemma}\label{lem:Z}
\[|Z|< |M\cup C|.\]
\end{lemma}
\begin{proof}
Assume that~$|Z|\geq |M\cup C|$. Then the subgraph 
induced by~$Z\cup M\cup C$ has more than 
$\frac{1}{2}4c|Z|$ edges and~$|Z\cup M\cup C|$ 
vertices. Hence its edge density is larger than 
$2c|Z|/(|Z\cup M\cup C|)\geq 2c|Z|/(2|Z|)= c$, 
contradicting our assumption on the edge density of 
\depthone minors of~$G$ (which includes its subgraphs). 
\end{proof}

Finally, we consider the image of the~$dom$-function inside the stars. 

\begin{lemma}\label{lem:insidestars}
\[|\bigcup_{m\in M\setminus D}\{u\in N_{E'}(m) : 
dom(u)\in (N_{E'}(m)\setminus (Y\cup Z))\}|\leq 
(2(t-1)+4)c|M|.\]
\end{lemma}
\begin{proof}
Fix some~$m\in M\setminus D$ and some
$u\in N_{E'}(m)$ with~$dom(u)\in N_{E'}(m)\setminus
(Y\cup Z)$. Because~$dom(u)\not\in Y$,~$d(u)$ is not
connected to a vertex of a different star, except possibly 
for vertices from~$M$. Because~$dom(u)\not\in Z$, it is 
however connected to at most~$4c$ vertices from~$M\cup C$. 
Hence it is connected to at most~$4c$ vertices 
from different stars. By Lemma~\ref{lem:edgestosamestar}, 
$dom(u)$ is connected to at most~$2c(t-1)$ vertices from
the same star. Hence the degree of~$dom(u)$ is at most 
$4c+2c(t-1)$. Because~$u$ preferred to choose~$dom(u)\in
N_{E'}(m)$ over~$m$ as its dominator, we conclude that~$m$ 
has at most~$4c+2c(t-1)$~$E'$-neighbors. Hence, in total there can 
be at most~$(2(t-1)+4)c|M|$ such vertices. 
\end{proof}

We are now ready to put together the numbers. 

\begin{lemma}\label{lem:mainlemma}
If~$G$ has edge density at most~$c$ and excludes~$K_{t,3}$ 
as \depthone minor, then the modified algorithm computes 
a~$6c^2t+(2t+5)c+4$ approximation for the 
minimum dominating set problem on~$G$. 
\end{lemma}
\begin{proof}
The set~$D$ has size at most~$(c+1)|M|$ according
to Lemma~\ref{thm:largeneighbourhood}. As~$M$ is 
a dominating set also with respect to the edges~$E'$,
it suffices to determine~$|\{ dom(u) : u\in (N_{E'}
[M\setminus D]\setminus N[D])\}|$. According to 
Corollary~\ref{crl:numedgesbetweendiamonds}, the 
set~$Y=\{u\in N_{E'}(m_i) :$ there is~$v\in N_{E'}(m_j)$, 
$i\neq j$ and~$\{u,v\}\in E(G)\}$ has size at most 
$6c^2t|M|$. In particular, there are at most so 
many vertices~$dom(u)\in N_{E'}(m_i)$ with 
$u\in N_{E'}(m_j)$ for~$i\neq j$. Clearly, at 
most~$|M|$~$dom$-edges can point to~$M$ and 
at most~$|M|$~$d$-edges point away from~$M$.
Together, this bounds the number of~$dom$-edges 
that go from one star to another star. According to 
Lemma~\ref{lem:Z}, there are only few vertices 
which are highly connected to~$M\cup C$, that is, 
the set~$Z=\{u \in N_{E'}(M\setminus D) : 
|N(u)\cap (M\cup C)|>4c\}$ satisfies~$|Z|< |M\cup C|$. 
We have~$|C|\leq 2c|M|$, as each~$C_m$ has size at 
most~$2c$. It remains to count the image of~$dom$ 
inside the stars which do not point to~$Y$ or~$Z$. 
According to Lemma~\ref{lem:insidestars}, this image has size 
at most~$(2(t-1)+4)c|M|$. In total, we hence find a set of size
\begin{align*}
 & (c+1)|M|+6c^2t|M|+2|M|+(2c+1)|M|+(2(t-1)+4)c|M|\\
 \leq & \; 6c^2t+(2t+5)c+4.
\end{align*}
\end{proof}

Our main theorem is now a corollary of 
Lemma~\ref{lem:exclude} and~\ref{lem:mainlemma}. 

\begin{theorem}\label{thm:main}
Let~$\mathcal{C}$ be a class of graphs of orientable 
genus at most~$g$ (non-orientable genus at most~$\tilde{g}$ resp.). 
The modified algorithm computes~$\Oof(g)$-approximation 
($\Oof(\tilde{g})$-approximation resp.) for the dominating set 
in~$\Oof(g)$ ($\Oof(\tilde{g})$ resp.) communication rounds. 
\end{theorem}

For the special case of planar graphs, our analysis shows that the algorithm 
computes a~$199$-approximation. This is not much worse 
than Lenzen et al.'s original analysis (130), however, off by a factor 
of almost~$4$ from Wawrzyniak's~\cite{better-upper-planar} improved
analysis (52).

\section{Improving the Approximation Factor with Preprocessing}\label{sec:improved}

 We now show that the constant approximation 
factors related to the genus~$g$, derived
in the previous section, can be improved, using a 
local preprocessing step. 

\begin{lemma}\label{lem:findk33}
Given a graph~$G$ and a vertex~$v\in V(G)$. 
A \emph{canonical subgraph 
$K_v\subseteq G$ of~$v$} on at most~$24$ vertices 
such that 
$v\in V(K_v)$ and~$K_{3,3}$ is a \depthone minor of~$K_v$,
can be computed locally in at most~$6$ communication rounds. 
\end{lemma}
\begin{proof}\label{alg:findk33}
The proof is constructive. Every \depthone minor of 
$K_{3,3}$ has diameter 
at most~$6$. Therefore, we can consider a subgraph 
$H=G[N^6(v)]$ and check whether there 
exists a minimal \depthone minor of~$K_{3,3}$ 
in~$H$ which includes~$v$ as a vertex. If this is the case, 
we 
output it as~$K_v$; otherwise we output the 
empty set. Furthermore,~$K_{3,3}$ has~$9$ 
edges and in a \depthone minor, each edge can be 
subdivided at most twice. Thus, in every minimal 
\depthone minor of~$K_{3,3}$, 
there are at most~$24$ vertices.
\end{proof}

To improve the approximation factor, we 
propose the following modified algorithm.

\begin{algorithm}[h!]
\caption{Dominating Set Approximation Algorithm for Graphs of Genus~$\leq g$}
\label{alg:modified-approx}
\begin{algorithmic}[1]
\vspace{2mm}
\STATE Input: Graph~$G$ of genus at most $g$%, integer~$c\geq \epsilon(G)$ 
\STATE \textbf{Run Phase~1 of Algorithm~\ref{alg:lenzen}~} 

\STATE~$(*$ \emph{Preprocessing Phase} ~$*)$

\STATE \textbf{for}~$v\in V-D$ (in parallel) \textbf{do}
\STATE ~~~~~ \textbf{compute~$K_v$ in~$G-D$ by Algorithm~\ref{alg:findk33}}

\STATE \textbf{for}~$i=1..g$ \textbf{do}
\STATE ~~\textbf{for}~$v\in V-D$ (in parallel) \textbf{do}
\STATE ~~~~~ \textbf{if}~$K_v\neq\emptyset$ \textbf{then}
\STATE ~~~~~~~~ \textbf{chosen : = true}
\STATE ~~~~~~~~ \textbf{for all~$u\in N^{12}(v)$}
\STATE ~~~~~~~~~~\textbf{if~$K_u\cap K_v \neq\emptyset$ and~$u < v$}\textbf{ then chosen := false}
\STATE ~~~~~~~~ \textbf{if (chosen = true) then }$D := D\cup V(K_v)$
\STATE \textbf{Run Phase~2 of Algorithm~\ref{alg:lenzen}~}
\end{algorithmic}
\end{algorithm}

\begin{theorem}\label{thm:modified}
Algorithm~\ref{alg:modified-approx} provides a~$24g+\Oof(1)$ MDS approximation 
for graphs of genus 
at most~$g$, and requires~$12g+\Oof(1)$ communication rounds.
\end{theorem}

\begin{proof}
The resulting vertex set is
clearly a legal dominating set.  
Moreover, as Phase~1 is unchanged,
we do not have to recalculate~$D$. 

In the preprocessing phase, if for two vertices~$u\neq v$ we choose
both~$K_u,K_v$, then they must be disjoint. 
Since the diameter of any \depthone minor of 
$K_{3,3}$ is at most~$6$, if two such canonical 
subgraphs intersect, the distance between 
$u,v$ can be at most~$12$. On the other hand, by 
Lemma~\ref{lem:decreasegenus}, there are at 
most~$g$ disjoint such models. So in
the \emph{preprocessing} phase, 
we select at most~$24g$ extra vertices for the dominating set.

Once the \emph{preprocessing} phase is finished, the remaining graph 
is locally embeddable. In fact, it does not have any 
$K_{3,3}$ \depthone minor. This is  
guranteed by repeating the process of~$K_{3,3}$ 
\depthone minor elimination~$g$ times in the 
preprocessing phase, and by the fact that in each run, 
if there are \depthone minors of~$K_{3,3}$'s, at
least one of them will be eliminated. Moreover, 
by 
Lemma~\ref{lem:decreasegenus}, there are no 
more than~$g$ such subgraphs. Hence, after 
repeating the procedure~$g$ times, 
the remaining graph is 
locally embeddable with constant edge density.

In order to compute the 
size of the set in the third phase, we 
can use the analysis of Lemma~\ref{lem:mainlemma}
for~$t=3$, which together with the first phase and preprocessing phase,
results in a~$24g+\Oof(1)$-approximation guarantee.

To count the number of communication rounds, note that 
the only change happens in the second phase. In that phase, 
in each iteration, we need~$12$ communication rounds to 
compute the~$12$-neighbourhood. Therefore, the number of 
communication rounds is~$12g + \Oof(1)$.
\end{proof}

This significantly improves the 
approximation upper bound of Theorem~\ref{thm:main}: namely from 
$4(6c^2+2c)g + \Oof(1)$, which, since~$c \le 3.01$ can be as high as 
$241g+\Oof(1)$ in sufficiently large graphs, to~$24g + \Oof(1)$,
at the price of~$12g$ extra communication rounds.

\section{A Logical Perspective}\label{sec:stoneage}

Interestingly, as we will elaborate in the following, a small modification of 
Algorithm~\ref{alg:lenzen} can be interpreted both from a distributed computing
perspective, namely as a local algorithm
of constant distributed time complexity, as well as from a logical perspective.

First order logic has atomic formulas of the form 
$x=y, x<y$ and~$E(x,y)$, where~$x$ and~$y$ are 
first-order variables. The set of first order formulas is 
closed under Boolean combinations and existential and 
universal quantification over the vertices of a graph. 
To define the semantics, we inductively define a satisfaction 
relation~$\models$, where for a graph~$G$, a formula 
$\phi(x_1,\ldots, x_k)$ and vertices~$v_1,\ldots, v_k
\in V(G)$,~$G\models\phi(v_1,\ldots, v_k)$ means 
that~$G$ satisfies~$\phi$ if the free variables~$x_1,\ldots, 
x_k$ are interpreted by~$v_1,\ldots, v_k$, respectively. 
The free variables of a formula are those not in the 
scope of a quantifier, and we write~$\phi(x_, 
\ldots , x_k)$ to indicate that the free variables
of the formula~$\phi$ are among~$x_1,\ldots, x_k$. 
For~$\phi(x_1,x_2)=x_1<x_2$, we have~$G\models
\phi(v_1,v_2)$ if~$v_1<v_2$ with respect to the 
ordering~$<$ of~$V(G)$ and for~$\phi(x_1,x_2)=
E(x_1,x_2)$ we have~$G\models\phi(v_1,v_2)$
if~$\{v_1,v_2\}\in E(G)$. The meaning of the 
equality symbol, the Boolean connectives, and the quantifiers is as expected.

A first-order formula~$\phi(x)$ with one free variable 
naturally defines the set~$\phi(G)=\{v\in V(G) : 
G\models\phi(v)\}$. We say that a formula~$\phi$ 
defines an~$f$-approximation to the dominating 
set problem on a class~$\CCC$ of graphs, if~$\phi(G)$ 
is an~$f$-approximation of a minimum dominating set for 
every graph~$G\in\CCC$. 

Observe that first-order logic is not able to count, in 
particular, no fixed formula can determine a neighbor of 
maximum degree in Line~14 of the algorithm. Observe 
however, that the only place in our analysis that refers to 
the dominator function~$d$ explicitly is Lemma~\ref{lem:insidestars}. 
The proof of the lemma in fact shows that we do not have to 
choose a vertex of maximal residual degree, but that it suffices 
to choose a vertex of degree greater than~$4c+2c(t-1)$ if such a 
vertex exists, or any vertex, otherwise. For every fixed class of 
bounded genus, this number is a constant, however, we have to 
address how logic should choose a vertex among the candidate 
vertices. For this, we assume that the graph is equipped with 
an order relation such that the formula can simply choose the 
smallest candidate with respect to the order. It is now easy
to see that the solution computed by the algorithm is first-order definable.

\section{Conclusion}\label{sec:FO}

This paper presented the first constant round, constant factor
local MDS approximation algorithm for the family of bounded genus graphs.
This result constitutes a major step forward
in the quest for understanding for which graph families such
local approximations exist. Besides the result itself, we believe 
that our analysis introduces several new techniques which may
be useful also for the design and analysis of local algorithms
for more general graphs, and also problems beyond MDS.

Moreover, this paper established an interesting connection between
distributed computing and logic, by presenting a local approximation
algorithm which is first-order logic definable. This also
provides an new perspective on the recently introduced
notion of stone-age distributed computing~\cite{stoneage}:
distributed algorithms making minimal assumptions on the power
of a node.

We believe that our work opens several very interesting directions
for future research. First, it remains open
whether
the local constant approximation result can be generalized to
sparse graphs beyond bounded genus graphs.
 Second, it will be interesting to extend our study of first-order definable approximations,
 an exciting new field.

%
%\section{Conclusion}\label{sec:conclusion}
%
%We explored approximation algorithms for the classic dominating set problem
%on sparse graphs.
%In particular, this paper assumed an interesting position
%at the intersection of distributed computing and logic:
%we presented a constant approximation algorithm which is
%not only \emph{local}, in the sense that it only requires a constant
%number of communication rounds,
%but it is also \emph{first-order definable}.
%To the best of our knowledge, it is a first witness
%for an algorithm belonging both to the constant distributed time complexity
%class as well as to the circuit complexity class
%$\mathrm{AC}^0$.
%
%We believe that our work opens several very interesting directions
%for future research. First, it remains open
%whether
%the local constant approximation result can be generalized to
%sparse graphs beyond local planar graphs and bounded genus graphs.
% Second, it will be interesting to extend our study of first-order definable approximations -- a very general and exciting topic which
% so far has not received much attention.
% In particular, it would be interesting to see whether
% first-order definable polynomial approximation schemes
% for dominating sets.

\newpage
{
\bibliographystyle{abbrv}
\bibliography{references}  % main.bib is the name of the Bibliography in this case
}

\end{document}